\documentclass[final,3p,times]{elsarticle}

\usepackage{latexsym}
\usepackage{amsmath}
\usepackage{amsthm}
\usepackage{amssymb}
\usepackage{enumerate}

\usepackage[normalem]{ulem}
\usepackage{verbatim}

\usepackage{subfigure}

\usepackage{colortbl}
\usepackage{tikz}
\tikzstyle{place}=[circle,draw=black,fill=gray!15,thick,inner sep=0pt,minimum size=5mm]
\tikzstyle{place1}=[circle,draw=black,fill=black!95,thick,inner sep=0pt,minimum size=1mm]
\tikzstyle{place3}=[circle,draw=black,fill=black,thick,inner sep=.3pt]
\usepackage{enumerate}

\newtheorem{theorem}{Theorem}

\newcommand{\rr}{\ensuremath{\mathbb{R}}}

\newcommand{\rob}{\textrm{Rob}}
\newcommand{\inc}{\textrm{Inc}}
\newcommand{\adv}{\textrm{Adv}}
\newcommand{\opt}{\textrm{Z}}

\newcommand{\pa}{K}
\newcommand{\bls}{\boldsymbol}
\newcommand{\bl}{\mathbf}

\usepackage{enumerate}
\usepackage{color}
\newcommand{\incfun}{\ensuremath{F}}
\newcommand{\robinc}{\text{RobInc}}

\usepackage{colortbl}
\usepackage{lscape}

\usepackage{subfigure}

\newcommand{\prob}[3]%
  { \vspace{.25cm}
    \noindent\fbox{\begin{minipage}{.98\textwidth}
      \textsc{#1}

      \vspace{.25cm}
      \begin{center}
        \begin{tabular}{lp{.8\textwidth}}
          Input:  & #2 \\
          Task: & #3
        \end{tabular}
      \end{center}
    \end{minipage}}

    \vspace{.25cm}}%

\newcommand{\probbox}[3]%
  { \vspace{.5cm}
    \noindent{\begin{minipage}{.98\textwidth}
      \textsc{\hspace{.5cm}#1}

      \vspace{.25cm}
      \begin{center}
        \begin{tabular}{lp{.8\textwidth}}
          Given:  & #2 \\
          Question: & #3
        \end{tabular}
      \end{center}
    \end{minipage}}

    \vspace{.25cm}}%

  { \vspace{.25cm}
    \noindent\fbox{\begin{minipage}{.98\textwidth}
      \textsc{#1}

      \vspace{.25cm}
      \begin{center}
        \begin{tabular}{lp{.8\textwidth}}
          Input:  & #2 \\
          Output: & #3
        \end{tabular}
      \end{center}
   \end{minipage}}
          \vspace{.25cm}
    \begin{enumerate}[(1)]}%
  {\end{enumerate}}

\journal{EJOR}

\begin{document}

\begin{frontmatter}



\title{Robust optimization with incremental recourse}


\author{Ebrahim Nasrabadi, James~B.~Orlin}
\address{Sloan School of Management and Operations Research Center, Bldg. E40-147, Massachusetts Institute of Technology, Cambridge, Massachusetts 02139\\
E-mail: \{nasrabad,jorlin\}$@$mit.edu}

\begin{abstract}
In this paper, we consider an adaptive approach to address optimization problems with uncertain cost parameters. Here, the decision maker selects an initial decision, observes the realization of the uncertain cost parameters, and then is permitted to modify the initial decision. We treat the uncertainty using the framework of robust optimization in which uncertain parameters lie within a given set. The decision maker optimizes so as to develop the best cost guarantee in terms of the worst-case analysis. The recourse decision is ``\emph{incremental}"; that is, the decision maker is permitted to change the initial solution by a small fixed amount. We refer to the resulting problem as the \emph{robust incremental} problem.  We study robust incremental variants of several optimization problems. We show that the robust incremental counterpart of a linear program is itself a linear program if the uncertainty set is polyhedral. Hence, it is solvable in polynomial time. We establish the NP-hardness for robust incremental linear programming for the case of a discrete uncertainty set. We show that the robust incremental shortest path problem is NP-complete when costs are chosen from a polyhedral uncertainty set, even in the case that only one new arc may be added to the initial path. We also address the complexity of several special cases of the robust incremental shortest path problem and the robust incremental minimum spanning tree problem. 
 \end{abstract}


\end{frontmatter}


\section{Introduction}
\label{sec:intro}
Researchers in the optimization community have developed a variety of approaches for addressing  problems of optimization under uncertainty. In general, there are two main approaches to address uncertainty in optimization models: stochastic optimization and robust optimization.  
The former approach treats the uncertainty in the data as random variables, thus giving a rich set of modeling tools.  These models typically lead to problems that are quite challenging to solve. It is also a practical challenge to determine probability distributions that accurately model the uncertainty.  For more information on stochastic optimization, see \cite{KallWallace94,BirgeLouveaux97,Shapiro09}. In contrast, robust optimization models the uncertainty in a deterministic manner. It assumes that the uncertain parameters come from known sets. It seeks solutions with the best worst-case cost guarantee. We refer the reader to \cite{BenTal04,BertsimasBrownCara11} for a discussion of the theory and applications of robust optimization. In this paper, we consider a modeling framework of robustness that allows the decision maker to adjust the initial solution by a bounded amount after the uncertain data is realized. In what follows, we provide a mathematical description of our problem and then review the related literature.

\paragraph{Problem description}
We consider optimization problems of the following form:
\begin{align}
\label{pro:Opt}
  \begin{aligned}
	& \min	&&\bl{c}^T \bl{x} \\
          & \text{~s.t.}   &&\begin{aligned}[t]
        							\bl{x}&\in  \mathcal{S},                
       				   \end{aligned}
\end{aligned}
\end{align}
where $\mathcal{S}\subseteq \rr^n$ denotes the feasible region and $\bl{c}\in \rr^n$ is the vector of cost parameters. For a fixed cost vector $\bl{c}$, we refer to the above problem as the \emph{nominal} problem. Notice that the vector~$\bl{c}$ is a column vector, and the superscript $T$ denotes the transpose operation. For clarity of the presentation, we will denote all vectors by bold small letters and all matrices by bold capital letters. 

We assume that the cost parameters are subject to uncertainty. We let $\mathcal{U}\subseteq \rr^n$ be the \emph{uncertainty set}, that is, the set of all possible cost vectors.  If $\bl{x}$ is a feasible solution, the objective function value for $\bl{x}$ in the robust optimization problem is 
$\max\{\bl{c}^T\bl{x} :\bl{c} \in \mathcal{U}\}$, which corresponds to the worst-case cost for solution $\bl{x}$. The mathematics can be viewed in terms of an \emph{adversary} who maliciously wants to increase the objective function value. The robust optimization problem is to find a minimax solution; that is, the feasible solution with the best worst-case cost guarantee. 

We further assume that the decision maker is able to make an incremental change in the solution~$\bl{x}$ after the uncertain cost parameters are revealed. We let $\mathcal{S}_{\bl{x}}:=\{\bl{y}\in \mathcal{S}:~ F(\bl{x},\bl{y})\leq \pa\}$ be the set of all possible choices, where $F(\bl{x},\bl{y})$ is a measure of the distance between solutions $\bl{x}$ and $\bl{y}$, and $\pa$ is a given upper bound on the total distance permitted. We refer to the set $\mathcal{S}_{\bl{x}}$ as  the \emph{incremental set} and to the function $F$ as the \emph{incremental function}. 

The goal is to find the best initial solution $\bl{x}$, assuming the worst-case cost scenario occurs and the decision maker is then allowed to transform the solution $\bl{x}$ into another solution $\bl{y}$ subject to the constraint that $\incfun(x,y)\leq \pa$.  This leads to the following \emph{robust incremental} optimization problem
\begin{align}
  \label{pro:IncRobOpt}
\opt_{\robinc}:=\min_{\bl{x}\in \mathcal{S}}~ \bl{d}^T\bl{x}+\max_{\bl{c}\in \mathcal{U}}~\min_{\bl{y}\in \mathcal{S}_{\bl{x}}} &\quad \bl{c}^T\bl{y}.
\end{align}
Notice that the objective function includes a term $\bl{d}^T\bl{x}$, where $\bl{d}$ is the cost vector for the initial decision. 

We illustrate our model with a shortest path problem under uncertainty. Suppose that a commuter needs to determine a route to be taken each day. Occasionally, roadwork causes certain routes to be unavailable or to have extra delays. In this case,  the robust incremental shortest path problem would be to select an initial route so as to minimize the worst-case travel given (1) the possible route delays are from a specified set $\mathcal{U}$, and (2) the modified route can vary from the initial route by a limited ``distance" $\pa$.  Another possible example is slowing or stopping an outbreak of a disease by imposing selected quarantines or limits on travel. As the disease progresses, some interdictions may become very costly, and the solution needs to be adjusted.  But in incremental optimization, adjustments to the original plan must be limited in scope. 

%
%
Problem~\eqref{pro:IncRobOpt} includes several interesting problems as special cases. If $\bl{d} = 0$ and if $\mathcal{S}_{\bl{x}}=\{\bl{x}\}$, it reduces to the following \emph{robust} problem: 
\begin{align}
  \label{pro:RobOpt}
\opt_{\rob}:=&\min_{\bl{x}\in \mathcal{S}}~ \max_{\bl{c} \in \mathcal{U}}~ \bl{c}^T\bl{x}.
\end{align}
If $\bl{d} = 0$ and $\mathcal{S}_{\bl{x}}=\mathcal{S}$, then Problem~\eqref{pro:IncRobOpt} reduces to the following \emph{adversarial} problem:  
\begin{align}
  \label{pro:AdvOpt}
\opt_{\adv}:=&\max_{\bl{c} \in \mathcal{U}}~\min_{\bl{y}\in \mathcal{S}_{\bl{x}}}~ \bl{c}^T\bl{y}.
\end{align}

The inner minimization problem in Problem~\eqref{pro:IncRobOpt} corresponds to the following \emph{incremental} problem:
\begin{align}
  \label{pro:IncOpt}
\opt_{\inc}:=&\min_{\bl{y}\in \mathcal{S}_{\bl{x}}} ~ \bl{c}^T\bl{y},
\end{align}
which addresses the case where $\bl{x}$ is fixed and the uncertain data is realized, and the goal is to make an improved subsequent decision. Notice that Problem~\eqref{pro:IncRobOpt}, in general, does not include the incremental problem~\eqref{pro:IncOpt}. However, for the case where $\mathcal{S}\subseteq \{0,1\}^n$, Problem~\eqref{pro:IncOpt}  becomes a special case of Problem~\eqref{pro:IncRobOpt} upon setting $d_i$ to be a very large value ($d_i\geq n C$ would be enough where $C:=\max\{c_i:~i=1,\ldots,n\}$) if $x_i=0$ and $d_i=0$, otherwise. Moreover, we require to set $\mathcal{U}=\{\bl{c}\}$. We then observe that the best initial solution for Problem~\eqref{pro:IncRobOpt} is $\bl{x}$. Since  $\mathcal{U}=\{\bl{c}\}$, the problem reduces to the incremental problem.

\paragraph{Literature review} 
Robust optimization has been well studied in the literature. In the case that the nominal problem is convex and the uncertainty set is polyhedral or ellipsoidal, tractable reformulations are known (see, e.g., \cite{BenElNem09,BertsimasBrownCara11} and the references therein). For 0--1 discrete optimization problems with uncertain cost parameters, Bertsiams and Sim \cite{BertsimasSim03} developed efficient algorithms for special cases in which the nominal problem is efficiently solvable.  Goetzmann \emph{et al.} \cite{GoetStilleTelha11} extended the results in \cite{BertsimasSim03} for integer programs with uncertain cost parameters and for integer programs with uncertainty in one or few constraints.


In robust optimization,  the decision maker must determine all the decisions simultaneously, and in particular before the realization of the uncertainty. This assumption leads to a single-stage optimization problem. 
Ben-Tal \emph{et al.} \cite{BenTal04} proposed a two-stage robust optimization modeling approach, the so-called \emph{adjustable robust} optimization, in which the decision maker makes two sets of decisions: one set before the uncertainty being realized and one set subsequently. They showed that the adjustable robust counterpart of a linear program with right hand side uncertainty is, in general, NP-hard. Subsequently, a number of researchers have attempted to obtain approximations to the adjustable robust optimization problem. Significant results have been obtained in designing approximation algorithms for adjustable robust optimization problems. We refer the reader to \cite{BenTal04,BertsimasBrownCara11} and the references therein for a review of the literature on adjustable robust optimization.  It is worthwhile pointing out that the robust incremental problem is a specific case of adjustable robust optimization and fits into the class of recoverable robustness. While previous research on adjustable optimization has mainly worked on problems with convex uncertainty sets, our main focus is on problems in which the uncertain costs reside within a discrete set.

In 2009, Liebchen \emph{et al.} \cite{Liebchen09} introduced a new concept of robust adjustable optimization, so-called \emph{recoverable robustness}, by allowing a limited recovery after the realization of uncertain parameters. They applied their modeling methodology to linear programming problems and provided an efficient algorithm in the case of right-hand side disturbances. Later, B{\"u}sing \cite{Busing12} specialized the concept of recoverable robustness to the shortest path problem and presented hardness results and approximation algorithms for different variants of the problem. We point out that the robust incremental problem fits into the class of recoverable robustness introduced by Liebchen \emph{et al.} \cite{Liebchen09} and is quite similar to the recoverable robust shortest path problem studied by B{\"u}sing \cite{Busing12}.

In robust optimization, the decision maker first chooses a solution, and then adversary selects the cost vector. In the adversarial problem,  the order of moves is reversed: the adversary first secrets the cost vector, and then
the decision maker selects a solution. Here, the adversary seeks a cost vector to maximize the optimal value of the solution chosen by the decision maker.  This generalizes \emph{the most vital variable problem}, which is to identify a variable whose cost increase leads to the largest increase in the optimal value (see, e.g., \cite{vanHoesel89}). This includes the problem of finding most vital arcs in network optimization, such as in the maximum flow problem \cite{Wollmer64,Wood93,RoysetWood07}, the minimum spanning tree problem \cite{LinChern93,Shen95} and the shortest path problem \cite{Bar-NoyKhullerSchieber95,IsraeliWood02,SanseverinoMarquez10}.  

In incremental optimization, it is assumed that an initial solution is given, and the aim is to make an incremental change in the solution that will result in the greatest improvement in the objective function. The incremental counterpart of several network flow problems was studied by \c{S}eref \emph{et al.} \cite{SerefAhujaOrlin09}. They presented a polynomial algorithm for the incremental minimum spanning tree problem. They showed that the incremental minimum cost flow problem can be solved in polynomial time using Lagrangian relaxation. They also considered two versions of the incremental minimum shortest path problem, where increments are measured via arc inclusions and arc exclusions. They presented a polynomial time algorithm for the arc inclusion version and showed that the arc exclusion version is NP-complete.

\paragraph{Modeling assumptions and notation} The analysis of Problem~\eqref{pro:IncRobOpt} depends on the set $\mathcal{S}$, the uncertainty set $\mathcal{U}$, and the incremental set $\mathcal{S}_{\bl{x}}$ (or equivalently the incremental function $F$). 
We study linear and integer optimization problems under different types of uncertainty sets and incremental functions.  
To model the uncertainty set, we assume that the nominal costs are given by $\bar{c}_i, i=1,\ldots,n$ and the cost of $i$-th cost coefficient can increase by at most $\hat{c}_i$. Thus, any possible realization of the $i$-th cost coefficient lies within the interval $[\bar{c}_i,\bar{c}_i+\hat{c}_{i}]$. Given a parameter $\Gamma>0$, we consider the following two uncertainty sets:
\begin{align}
\label{eq:U1}
\mathcal{U}_1&:=\{\bl{c}=\bar{\bl{c}}+\bls{\delta} :~ \bl{0}\leq \bls{\delta} \leq \hat{\bl{c}},~\sum_{i=1}^{n}\delta_{i}\leq \Gamma\},\\
\label{eq:U2}
\mathcal{U}_2&\begin{aligned}[t]
:=\{\bl{c}=\bar{\bl{c}}+\bls{\delta}\cdot \hat{\bl{c}}  :~ \bls{\delta}\in \{0,1\}^n, \sum_{i=1}^{n}\delta_i\leq \Gamma\},
 \end{aligned}
\end{align}
where $\bls{\delta}\cdot \hat{\bl{c}}$ is the doc product of the vectors $\bls{\delta}$ and $\hat{\bl{c}}$.
The uncertainty set $\mathcal{U}_1$ models the case where the total change permitted in the  cost coefficients is limited by $\Gamma$, whereas in the uncertainty set $\mathcal{U}_1$, the cost of at most $\Gamma$ coefficients is allowed to increase from their nominal costs (here $\Gamma$ is assumed to be integer). The extreme points of the set $\mathcal{U}_1$ are the cost vectors in $\mathcal{U}_2$. Hence, the two sets lead to equivalent robust optimization problems. But this is not the case for the adversarial problem and the robust incremental problem as the two sets lead to different problems.
%
%

%
%

Because robust incremental problems involve three stages of decisions, it is unlikely that they will be in the class NP (assuming that the polynomial time hierarchy does not collapse.)  
For example, consider the special case of the robust incremental problem (\ref{pro:IncRobOpt}) in which $S$ denotes the feasible region of a combinatorial optimization problem, and
where $\mathcal{U} = \mathcal{U}_1$, and $S_x = \{y: \sum_{i=1}^{n}|x_i-y_i| \le K\}$ for some specified $K$.  
Suppose that the nominal optimization problem is solvable in polynomial time. 
In this case, the decision variant of the third stage problem, i.e., the incremental problem, is in the class P.  (The decision variant is a yes-no variant of the optimization problem.)  The decision variant of the second stage decision, i.e., the adversarial problem, is in the class NP, and the decision variant of the robust incremental problem is in the class $\Sigma^p_2$, which is one level higher than NP in the polynomial time hierarchy (see, e.g., \cite{meyer1972equivalence,stockmeyer1976polynomial}).
 If the decision variant of the third stage problem, i.e., the incremental problem, is in NP, then the adversarial problem  is in the class
$\Sigma^p_2$.  And the decision variant of the robust incremental problem is in the class $\Sigma^p_3$, which is one level higher than $\Sigma^p_2$ on the polynomial time hierarchy. 

In this paper, we address problems in which the nominal problem is in the class P.  
 In particular, we consider robust incremental variants of the following three types of problems: (1) linear programming problems, (2)  shortest path problems, and  (3) minimum spanning tree problems.  In cases in which the adversarial problem is NP-hard, it remains an open question as to whether the decision variant of the robust incremental problem  $\Sigma^p_2$-complete.

We consider three types of incremental functions.  For linear programming, we restrict attention to the $L_1$--distance  function, that is, $F(\bl{x},\bl{y})= \sum_{i=1}^{n}|x_i-y_i|$.  For the shortest path problem, we consider two additional incremental functions. In measuring the distance of a path $P$ from a given path $P^0$, we consider the arc inclusion function $|P\setminus P^0|$ and the arc exclusion function $|P^0\setminus P|$.  We note that the $L_1$--distance in this case is equivalent to the arc symmetric-difference $|P^0\oplus P|$. 
For the minimum spanning tree problem, the complexity for all three metrics is equivalent, and we measure the distance of a tree $T$ from a given tree $T^0$ by the arc inclusion  function $ |T\setminus T^0|$. 

\paragraph{Our contribution} 
Our primary contributions are as follows:
\begin{enumerate}
\item We show that the robust incremental counterpart of a linear program with respect to uncertainty set $\mathcal{U}_1$ is transformable into a linear program and consequently solvable in polynomial time.  In the case that the uncertainty set is $\mathcal{U}_2$, we show that robust incremental linear programming is NP-hard.  In particular, we show that the adversarial minimum cost flow problem is NP-hard with respect to the discrete uncertainty set $\mathcal{U}_2$.
\item We show that the robust incremental shortest path problem is NP-hard even in the special case that the uncertainty set is $\mathcal{U}_1$ and only one new arc may be added to the initial path.
\item We prove that the adversarial shortest path problem is solvable in polynomial time with respect to the uncertainty set $\mathcal{U}_1$ and the arc inclusion function, while the robust incremental problem is NP-hard. We also show that the incremental shortest path problem  is NP-hard with respect to the arc symmetric-difference function. 
\item We also consider an adversarial variant of the minimum spanning tree problem. We show that the adversarial minimum spanning tree problem is solvable in polynomial time when the cost parameters lie in the uncertainty set $\mathcal{U}_1$.
\end{enumerate}

%

\begin{table}[t]
\caption{\label{tbl:main}Complexity of different problems}
\centering
 \begin{tabular}{|c|c|c|c|c|}
\hline\hline Problem/Unc/Dist&$\opt_{\rob}$ &$\opt_{\inc}$ & $\opt_{\adv}$& $\opt_{\robinc}$\\
\hline LP/$\mathcal{U}_1$/$L_1$	&P ~\cite{Ben-TalNemirovski99,BertsimasSim04} 	&P~\cite{SerefAhujaOrlin09}& \bf{P} 	& \bf{P}\\
\hline LP/$\mathcal{U}_2$/$L_1$	&P ~\cite{Ben-TalNemirovski99,BertsimasSim04}	&P~\cite{SerefAhujaOrlin09}&\bf{NPC} 	& \bf{NPH}\\
\hline SP/$\mathcal{U}_1$/Incl		&P ~\cite{BertsimasSim03} 	&P~\cite{SerefAhujaOrlin09}&\bf{P}	& {\bf{NPC}}\\
\hline SP/$\mathcal{U}_1$/Excl		&P ~\cite{BertsimasSim03} &NPC~\cite{SerefAhujaOrlin09}&NPC~\cite{SerefAhujaOrlin09}&NPH~\cite{SerefAhujaOrlin09}\\
\hline SP/$\mathcal{U}_1$/Sym		&P ~\cite{BertsimasSim03} 	&\bf{NPC} &\bf{NPH} 	& \bf{NPH}\\
\hline SP/$\mathcal{U}_2$/Incl		&P ~\cite{BertsimasSim03} 	&P~\cite{SerefAhujaOrlin09}&NPC~\cite{Bar-NoyKhullerSchieber95}	& {\bf{NPH}}~\cite{Bar-NoyKhullerSchieber95,Busing12}\\
\hline SP/$\mathcal{U}_2$/Excl		&P ~\cite{BertsimasSim03} 	&NPC~\cite{SerefAhujaOrlin09} 	&NPC~\cite{Bar-NoyKhullerSchieber95}& NPH~\cite{Bar-NoyKhullerSchieber95}\\
\hline SP/$\mathcal{U}_2$/Sym		&P ~\cite{BertsimasSim03} 	&\bf{NPC}	&NPC~\cite{Bar-NoyKhullerSchieber95}& {\bf{NPH}}~\cite{Bar-NoyKhullerSchieber95}\\
\hline MST/$\mathcal{U}_1$/Incl	&P ~\cite{BertsimasSim03}	&P~\cite{SerefAhujaOrlin09}&\bf{P}	& -- --\\
\hline MST/$\mathcal{U}_2$/Incl	&P ~\cite{BertsimasSim03} 	&P~\cite{SerefAhujaOrlin09}	&NPC~\cite{LinChern93}	& NPH~\cite{LinChern93}	\\
\hline\hline
\end{tabular}
\\[2mm]
\raggedright
\end{table}

We refer to Table \ref{tbl:main} for a summary of our contributions as well as a summary of known results from the literature. The input for the problem is given as a triple.  The first term of the triple is LP, SP, or MST indicating linear programs, shortest paths or minimum spanning trees.  The second term of the triple is $\mathcal{U}_1$ or $\mathcal{U}_2$, indicating the type of uncertainty set.  The third term of the triple is the type of distance metric.  It is either $L_1$ or it is an abbreviation for set inclusion, set exclusion, or set symmetric difference.
We indicate the complexity of the decision version of the problem as P (polynomially solvable), NPC (NP-complete) or NPH (NP-hard).  If the complexity is stated with a bold font, it refers to results established in this paper.  For SP/$\mathcal{U}_2$/Incl,  B{\"u}sing \cite{Busing12} shows that the problem is NP-hard for constant $K\geq 1$. Here, we show that the problem is NP-hard, even for $\Gamma=K=1$.  For SP/$\mathcal{U}_1$/Incl, we prove that the problem is NP-hard and inapproximable within within a factor of 2, even for $\Gamma=K=1$
Of the 28 different problems \footnote{Notice that some problems in the table are equivalent. For example, LP/$\mathcal{U}_1$/$L_1$ is equivalent to LP/$\mathcal{U}_2$/$L_1$ for the robust problem. In total, there are 28 different problems in Table \ref{tbl:main}.} summarized in Table \ref{tbl:main}, only one of the problems remains open.

\section{Robust-Incremental Linear Programming}
\label{sec:LP}
In this section, we consider the robust incremental optimization model for linear programming. More precisely, we assume that the feasible region $\mathcal{S}$ is given by non-negativity constraints as well as a number of linear equalities, that is,  
\begin{align}
\label{eq:S-LP}
\mathcal{S}:=\{\bl{x}\in \rr^{n}_{+}:~ \bl{A}\bl{x}=\bl{b}\},
\end{align}
where $\bl{A}$ is an $m\times n$ matrix and $\bl{b}$ is an $m$-vector of the right-hand side parameters. 

We first consider the uncertainty set $\mathcal{U}_1$, where the total change in the cost parameters is bounded by $\Gamma$. We further assume that the total change permitted in the initial solution is bounded by $\pa $. The incremental set is given by
\begin{align}
\label{eq:S_x-LP}
\mathcal{S}_{\bl{x}}:=\{\bl{y}\in \mathcal{S} :~ \sum_{i=1}^{n}|x_i-y_i|\leq \pa \}.
\end{align}

We show that the resulting robust incremental optimization problem may be represented as a linear programming problem, and hence can be solved efficiently. With respect to a solution $\bl{x}\in \mathcal{S}$ and a cost vector $\bl{c} \in \mathcal{U}$,  we define $\opt_{\inc}(\bl{x},\bl{c})$ to be the optimal value of the incremental problem, that is, 
\begin{align}
\label{pro:Inc-LP(NLP)}
  \begin{aligned}
	\opt_{\inc}(\bl{x},\bl{c}) := 	& \min	&&\bl{c}^T\bl{y} \\
          		  		& \text{~s.t.}   &&\begin{aligned}[t]
        							Ay  &=\bl{b}, \\
							\sum_{i=1}^{n}|x_i-y_i|&\leq \pa ,\\
							\bl{y}     &\geq \bl{0}.          
  \end{aligned}
\end{aligned}
\end{align}
We introduce two nonnegative variables $z_i^{+}, z_i^{-}$, and set $y_i-x_i=z_i^{+}-z_i^{-}$  for $i=1,\ldots,n$. The variables $z_i^{+}$ and $z_i^{-}$ correspond to the increase or decrease in $x_i$, respectively. There is always an optimal solution in which either $z_i^{+}=0$ or $z_i^{-}=0$, in which case $| z_i^{+}-z_i^{-}|$ = $z_i^{+}+z_i^{-}$.
Hence, $\opt_{\inc}(\bl{x},\bl{c})$ is equivalent to the following linear program:
\begin{align}
\label{pro:Inc-LP(LP)}
  \begin{aligned}
	\opt_{\inc}(\bl{x},\bl{c}) = & \min	&&\bl{c}^T\bl{x}+\bl{c}^T(\bl{z}^{+}-\bl{z}^{-}) \\
          		  		& \text{~s.t.}   &&\begin{aligned}[t]
        							\bl{A}(\bl{z}^{+}-\bl{z}^{-})  &= \bl{0}, \\
							\sum_{i=1}^{n}(z_i^{+}+z_i^{-})&\leq \pa ,\\
							\bl{z}^{-}&\leq \bl{x},\\
							\bl{z}^{+},\bl{z}^{-}     &\geq \bl{0}.          
  \end{aligned}
\end{aligned}
\end{align}
The dual of the above problem is the following linear program:
\begin{align}
\label{pro:Inc-LP*(LP)}
  \begin{aligned}
	\opt_{\inc}(\bl{x},\bl{c}) = & \max	&&\bl{c}^T\bl{x}-\alpha\pa -\bl{x}^T\bl{v} \\
          		  		& \text{~s.t.}   &&\begin{aligned}[t]
        							\bl{w}^T\bl{A}-\alpha \bl{1}  &\leq \bl{c}, \\
							-\bl{w}^T\bl{A}-\alpha \bl{1}-\bl{v}  &\leq -\bl{c}, \\
							\bl{v}     &\geq \bl{0},\\         
							\alpha     &\geq 0,         
  \end{aligned}
\end{aligned}
\end{align}
where $\bl{w}$ is an $m$-vector, $\bl{v}$ is an $n$-vector, and $\bl{1}$ is an $m$-vector each of whose entries is one. Then $\bl{w}$ is the vector of dual variables corresponding to the first set of constraints, $-\alpha$ is a  dual variable corresponding to the second constraint, and $\bl{v}$ is the vector of dual variables corresponding to the third set of constraints in Problem~\eqref{pro:Inc-LP(LP)}.

We now define $\opt_{\adv}(\bl{x}):=\max_{\bl{c} \in \mathcal{U}_1} \opt_{\inc}(\bl{x},\bl{c}) $ to be the optimal value of the adversarial problem with respect to a given solution $\bl{x}$. Following the above discussion,  $\opt_{\adv}(\bl{x})$ can be expressed as follows:
\begin{align}
\label{pro:Adv-LP}
  \begin{aligned}
	\opt_{\adv}(\bl{x}) =& \max	&&\bar{\bl{c}}^T\bl{x}+\bl{\delta}^T\bl{x}-\alpha\pa -\bl{x}^T \bl{v}  \\
          		  		& \text{~s.t.}   &&\begin{aligned}[t]
        							\bl{w}^T\bl{A}-\alpha \bl{1}-\bl{\delta}  &\leq \bar{\bl{c}}, \\
							-\bl{w}^T\bl{A}-\alpha \bl{1}-\bl{v}+\bl{\delta} &\leq -\bar{\bl{c}}, \\
							\sum_{i=1}^{n}\delta_i&\leq \Gamma,\\
							0\leq \bl{\delta} &\leq \hat{\bl{c}},\\
							\bl{v}     &\geq \bl{0},\\         
							\alpha     &\geq 0.        
  \end{aligned}
\end{aligned}
\end{align}

We once again take the dual and obtain:
\begin{align}
\label{pro:Adv-LP*}
  \begin{aligned}
	\opt_{\adv}(\bl{x}) =  & \min	&&\bar{\bl{c}}^T(\bl{x}+\bl{z}^{+}-\bl{z}^{-})+\beta\Gamma+\hat{\bl{c}}^T\bl{q}  \\
          		  		& \text{~s.t.}   &&\begin{aligned}[t]
       							\bl{A}(\bl{z}^{+}-\bl{z}^{-})  &= 0, \\
							\sum_{i=1}^{n}(z_i^{+}+z_i^{-})&\leq \pa ,\\
        							-\bl{z}^{+}+\bl{z}^{-}+\beta \bl{1}+\bl{q}  &\geq \bl{x}, \\
							\bl{z}^{-}&\leq \bl{x},\\
							\bl{z}^{+},\bl{z}^{-},\bl{q}     &\geq \bl{0},\\          
							\beta     &\geq 0.          
  \end{aligned}
\end{aligned}
\end{align}

The goal of Problem~\eqref{pro:IncRobOpt} is find a solution $\bl{x}\in \mathcal{S}$ with minimum value $\opt_{\adv}(\bl{x})$. This establishes the following theorem.

\begin{theorem}
Suppose that the sets $\mathcal{S}$, $\mathcal{U}$ and $\mathcal{S}_{\bl{x}}$ are given by \eqref{eq:S-LP}, \eqref{eq:U1}, and \eqref{eq:S_x-LP}, respectively. Then, the robust incremental optimization Problem~\eqref{pro:IncRobOpt} may be formulated as the following linear programming problem:
\begin{align}
  \begin{aligned}
& \min	&&\bl{d}^T\bl{x}+\bar{\bl{c}}^T(\bl{x}+\bl{z}^{+}-\bl{z}^{-})+\beta\Gamma+\hat{\bl{c}}^T\bl{q}   \\
          		  		& \text{~s.t.}   &&\begin{aligned}[t]
       							\bl{A}\bl{x}&=\bl{b},\\
       							\bl{A}(\bl{z}^{+}-\bl{z}^{-})  &= 0, \\
							\sum_{i=1}^{n}(z_i^{+}+z_i^{-})&\leq \pa ,\\
        							-\bl{z}^{+}+\bl{z}^{-}+\beta \bl{1}+\bl{q}  &\geq \bl{x}, \\
							\bl{z}^{-}&\leq \bl{x},\\
							\bl{z}^{+},\bl{z}^{-},\bl{q}     &\geq \bl{0},\\          
							\bl{x}     &\geq 0,  \\        
							\beta     &\geq 0.          
  \end{aligned}
\end{aligned}
\end{align}
Therefore, Problem~\eqref{pro:IncRobOpt} can be solved in polynomial time.
\end{theorem}

We next turn our attention to the  uncertainty set $\mathcal{U}_2$ and show that the robust incremental counterpart of a linear program in NP-hard. To this end, we considet the minimum cost flow problem. In this problem, we are given a directed graph $G=(N,A)$ with \emph{node set} $N$ and \emph{arc set} $A\subseteq N\times N$. Each arc $(i,j)\in A$ has an associated \emph{capacity}~$u_{ij}$ and an associated \emph{cost} $c_{ij}$.  The supply/demand of node $i$ is $b_i$. We assume that $\sum_{i\in N} b_i=0$. 

A \emph{flow} is a vector $\bl{x}\in \rr^{|A|}_+$ that assigns a nonnegative value $x_{ij}$ to arc $(i,j)$. We refer to $x_{ij}$ as the \emph{flow} on arc $(i,j)$. We require that a flow $\bl{x}$ obeys the \emph{capacity constraints} $x_{ij}\leq u_{ij}$,  $(i,j)\in A$ and satisfies the following \emph{flow conservation constraints}:
\begin{align*}
\sum_{j:(i,j)\in A}x_{ij}-\sum_{j:(j,i)\in A}x_{ji} &=b_i && \forall i\in N.
\end{align*}

The \emph{cost} of a flow $\bl{x}$ is given by $\sum_{(i,j)\in A} c_{ij}x_{ij}$. The \emph{minimum cost flow problem} is to find a feasible flow of minimum cost. This problem can be stated as follows:
\begin{align}
\label{pro:MCFP}
  \begin{aligned}
	&\min	&&\bl{c}^T\bl{x} \\
         &   \text{~s.t.}   &&\begin{aligned}[t]
        							\bl{N}\bl{x}  &=\bl{b},                \\
		  					 \bl{0}\leq \bl{x}               &\leq \bl{u},
       				   \end{aligned}
\end{aligned}
\end{align}
where $\bl{N}$ is the node-arc adjacency matrix of the graph $G$, $\bl{u}$ is the vector of arc capacities, and $\bl{b}$ is the vector of supplies or demands.

We suppose that arc costs are uncertain.  The cost of arc $(i,j)$ is in the interval $[\bar{c}_{ij},\bar{c}_{ij}+\hat{c}_{ij}]$. 
We assume that the total change permitted in the incremental stage is bounded by $\pa $. 
Given a flow $\bl{x}$, the \emph{adversarial} minimum cost flow problem is to find a cost vector $\bl{c}\in \mathcal{U}_1$ that maximizes the optimal value $\opt_{\inc}(\bl{x},\bl{c})$ of the incremental problem. 

In the case that the uncertainty set is  $\mathcal{U}_1$, the results of the previous section show that the robust incremental min cost flow problem can be formulated as a linear programming problem, and is thus solvable in polynomial time. 

We next show that the adversarial minimum cost flow problem is NP-hard under the uncertainty set $\mathcal{U}_2$. 

The adversarial problem is
to increase the costs of at most $\Gamma$ arcs such as to maximize the minimum cost flow. This problem can be formulated as follows:
 \begin{align}
  \label{pro:IR-MCF-Arcs}
\begin{aligned}
 && \max_{\bls{\delta}\in \Theta}~\min &\quad  \sum_{(i,j)\in A} (\bar{c}_{ij}x_{ij}+\delta_{ij}\hat{c}_{ij}x_{ij})\\
      &&        \text{~s.t.}&\quad\begin{aligned}[t]
         							\bl{N}\bl{x}  &=\bl{b},                \\
		  					 \bl{0}\leq \bl{x}               &\leq \bl{u},
       					\end{aligned}
\end{aligned}
\end{align}
where $\Theta:=\{\bls{\delta}\in \{0,1\}^{|A|}:~ \sum_{(i,j)\in A}\delta_{ij}\leq \Gamma\}$. 


We show that this problem is NP-hard by a transformation from the \emph{network interdiction} problem, which is defined as follows:  reduce the value of a maximum flow in a given network as much as possible by removing $\Gamma$ arcs. Wood \cite{Wood93} shows that this problem is NP-hard by a transformation from the clique problem.

\begin{theorem}
Under the uncertainty set  $\mathcal{U}_2$, the adversarial minimum cost flow problem is NP-hard.
\end{theorem}
\begin{proof}
Consider the following decision version of the network interdiction  problem.  We are given a capacitated network $G=(N,A)$ with a source $s\in N$ and sink $t\in N$, an integer $\Gamma$, and an integer $k$.  The problem is to determine whether there exists a set of $\Gamma$ arcs so that after removing these arcs the maximum flow value is at most $k-1$. We now construct an instance of the adversarial minimum cost flow problem as follows. Suppose that we want to ship $k$ units of flow from $s$ to $t$ in the given network. We let $b_s=k$, $b_t=-k$, and $b_i=0$ for all other nodes. Moreover,  assume that the nominal costs $\bar{c}_{ij}$ are all zero and an adversary can increase the costs of at most $\Gamma$ arcs by one, that is,  $\bar{c}_{ij}=0$ and $\hat{c}_{ij}=1$ for all arcs $(i,j)\in A$.  

We next show that the network interdiction flow problem is a ``Yes''-instance if and only if the optimal value of the corresponding adversarial minimum cost flow problem is strictly positive. Suppose that the network interdiction flow problem is a ``Yes''-instance. Thus, there are $\Gamma$ arcs so that one cannot ship $k$ units of flow from $s$ to $t$ without sending flow on at least one of these arcs. This implies that if the adversary increases the cost of these $\Gamma$ arcs by one, then the cost for sending $k$ units of flow will be strictly positive.  Conversely, suppose that the optimal value of the adversarial minimum cost flow problem is strictly positive.   Let $S$ denote the set of $\Gamma$ arcs whose cost was increased from 0 to 1.  Then any flow of $k$ units from $s$ to $t$ must use at least one of these arcs. This completes the proof of the theorem. 
\end{proof}

The above theorem implies that under the uncertainty set $\mathcal{U}_2$, robust incremental network flows (and thus robust incremental linear programming) is NP-hard.  It is an open question whether the decision variant of robust incremental network flows  is in the class NP.  It is also open as to whether the problem is $\Sigma^p_2$-complete.

\section{Robust-Incremental Shortest Path Problem}
\label{sec:SP}
Here, we study the shortest path problem from a robust incremental viewpoint.  Let $G=(N,A)$ be a directed graph with a \emph{source} $s\in N$ and a \emph{sink} $t\in N$, and let $c_{ij}$ denote the cost (length) of arc $(i,j)\in A$. In the \emph{shortest path problem}, we seek a path of least cost from $s$ to $t$.  

We suppose that for each arc $(i,j)$, the cost $c_{ij}$ is uncertain and can vary within the interval $[\bar{c}_{ij},\bar{c}_{ij}+\hat{c}_{ij}]$. If the decision maker chooses a path $P^0$ in the first stage and observes the realization $\bl{c}$ of the cost vector, she is allowed to build a new path $P$ in the second stage whose distance from $P^0$ is not more than $\pa$ for some specified integer $K$.  We measure the distance of a path $P$ from the path $P^0$ via the three different incremental functions: $|P\setminus P^0|$, $|P^0\setminus P|$, and $|P\oplus P^0|$.  

If $K \ge 2n-2$ and if the uncertainty set is $\mathcal{U}_2$, the \emph{adversarial} shortest path problem is the problem of maximizing the length of the least cost path from $s$ to $t$ by increasing the cost of at most $\Gamma$ arcs. 

By setting $\hat{c}_{ij}$ to be very large, the adversarial shortest path problem reduces to the problem of determining the $\Gamma$ most vital arcs, i.e., those $\Gamma$ arcs whose removal results in the greatest increase in the length of the shortest path from $s$ to $t$. The latter problem is known to be NP-hard and inapproximable by a constant factor better than 2 \citep[see][]{Bar-NoyKhullerSchieber95}. This implies that the robust incremental shortest path problem is NP-hard with respect to all three incremental functions. In the rest of this section, we examine the robust incremental shortest path problem with respect  the uncertainty set $\mathcal{U}_1$, unless mentioned otherwise, and under the incremental functions $|P\setminus P^0|$ and $|P^0\oplus P|$.

\subsection{The incremental function $|P\setminus P^0|$}
Here, we assume that one can build a new path by adding at most $K$ new arcs to the path in the incremental stage. For this case, \c{S}eref \emph{et al.} \cite{SerefAhujaOrlin09} present a polynomial time algorithm for solving the incremental shortest path problem. We show that the adversarial problem is also solvable in polynomial time.  And we show that the robust incremental problem is NP-hard, even for the special case where $\Gamma=K=1$.

Let  $\opt_{\inc}(P^0,\bl{c})$ be the optimal value of the incremental optimization problem on $G = (N, A)$ with respect to a given path $P^0$ and a vector $\bl{c}$ of costs.  We next transform the incremental optimization problem on $G$ into a shortest path problem on a time expanded network $G^* = (N^*, A^*)$.  We create the time expanded network $G^*$  as follows.  
For each node $i \in N$ there are $K+1$ copies of the node in  $N^*$.  The copies are denoted as $i_{k}$, for $k=0,\ldots, K$.
  
There are three subsets of arcs.  For every arc $(i, j) \in P^0$, there are arcs 
$(i_{k}, j_{k})$, for $k=0,\ldots, K$.   For every arc $(i, j) \in A\setminus P^0$, there are arcs $(i_{k}, j_{k+1})$, for $k=0,\ldots, K-1$.  Finally, for every node $i \in N$, there are arcs $(i_{k}, i_{k+1})$, for $k=0,\ldots, K-1$.  Arcs in the first two sets have the same cost as in $G$.  Arcs in the third set have a cost of 0.  {The construction of the time expanded network is illustrated in Figure \ref{fix:EimeExp} for a simple network.} 

We see that the transformation is valid as follows.  A feasible path $P$ in the original network with $|P\setminus P^0|\leq K$ will induce a path $Q$ from $s_0$ to $t_K$ in the time expanded network so that the cost of $P$ and $Q$ are the same.  Similarly, a feasible path $Q$ from $s_0$ to $t_K$ in the time expanded network induces a path $P$ in $G$ whose cost is the same as that of $Q$, and such that  $|P\setminus P^0|\leq K$.

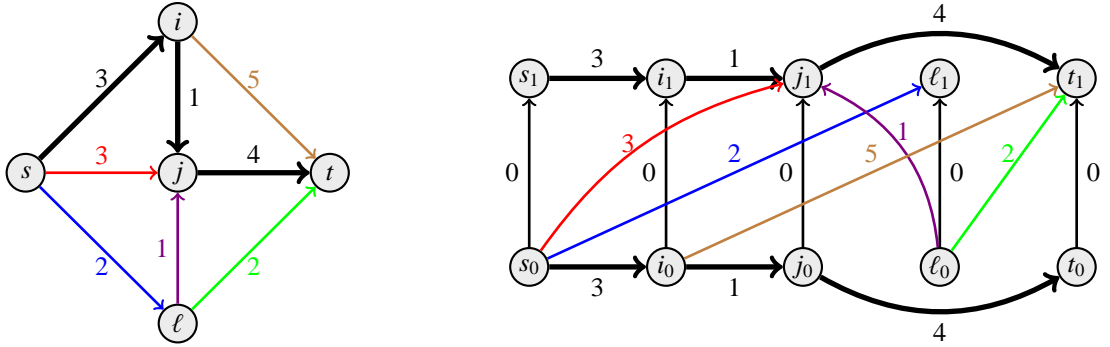
\begin{figure}[t]
\begin{minipage}{0.5\textwidth}
\begin{center}
\begin{tikzpicture}[inner sep=1mm]
  \node (1) at ( -2,0) [place] {$s$};
  \node (2) at ( 0,2) [place] {$i$};
  \node (3) at ( 0,-2) [place] {$\ell$};
  \node (4) at ( 2,0) [place] {$t$};
  \node (5) at ( 0,0) [place] {$j$};
\begin{scope}[color=black,line width=1pt]
  \draw [->,line width=2pt] (1) -- (2) node [above,text centered,midway]{3};
  \draw[blue,->] (1) -- (3) node [below,text centered,midway]{2};
  \draw [->,line width=2pt] (2) -- (5) node [right,text centered,midway]{1};
  \draw [brown,->] (2) -- (4) node [above,text centered,midway]{5};
  \draw [green,->] (3) -- (4) node [below,text centered,midway]{2};
  \draw [violet,->] (3) -- (5) node [left,text centered,midway]{1};
  \draw [red,->] (1) -- (5) node [above,text centered,midway]{3};
  \draw [->,line width=2pt] (5) -- (4) node [above,text centered,midway]{4};
\end{scope}
\end{tikzpicture}
\end{center}
\end{minipage}%
\begin{minipage}{0.5\textwidth}
\begin{center}
\begin{tikzpicture}
  \node (1) at ( 0,0) [place] {$s_0$};
  \node (2) at ( 1.8,0) [place] {$i_0$};
  \node (3) at ( 3.6,0) [place] {$j_0$};
  \node (4) at ( 5.4,0) [place] {$\ell_0$};  
  \node (5) at ( 7.2,0) [place] {$t_0$};
  \node (6) at ( 0,2.5) [place] {$s_1$};
  \node (7) at ( 1.8,2.5) [place] {$i_1$};
  \node (8) at ( 3.6,2.5) [place] {$j_1$};
  \node (9) at ( 5.4,2.5) [place] {$\ell_1$};
  \node (10) at (7.2,2.5) [place] {$t_1$};
\begin{scope}[color=black,line width=1pt]
  \draw [->] (1) -- (6) node [left,text centered,midway]{0};
  \draw [->] (2) -- (7) node [left,text centered,midway]{0};
  \draw [->] (3) -- (8) node [left,text centered,midway]{0};
  \draw [->] (4) -- (9) node [right,text centered,midway]{0};
  \draw [->] (5) -- (10) node [right,text centered,midway]{0};
  \draw [->,line width=2pt] (1) -- (2) node [below,text centered,midway]{3};
  \draw [->,line width=2pt] (2) -- (3) node [below,text centered,midway]{1};
  \draw [->,line width=2pt] (6) -- (7) node [above,text centered,midway]{3};
  \draw [->,line width=2pt] (7) -- (8) node [above,text centered,midway]{1};
  \draw [red,->] (1)  to [bend left=20] node [left,text centered,midway]{3}  (8);
  \draw [blue,->] (1)  to [bend right=0]  node [above,text centered,midway]{2} (9);
  \draw [green,->] (4) -- (10) node [above,text centered,midway]{2};
  \draw [violet,->] (4)  to [bend right=30]  node [above,text centered,midway]{1} (8);
  \draw [brown,->] (2) -- (10) node [above,text centered,midway]{5};
  \draw [->,line width=2pt] (3) to [bend right=30]  node [below,midway]{4}(5);
  \draw [->,line width=2pt] (8) to [bend left=30]  node [above,text centered,midway]{4}(10);
\end{scope}
\end{tikzpicture}
\end{center}
\end{minipage}%
\caption{\label{fix:EimeExp} {On the left hand side, a network $G$ is shown. Let $P^0=s,i,j,t$ and let $K=1$. On the right hand side, the corresponding time expanded network $G^*$ is depicted. The number on the arc indicates the cost. }}
\end{figure}

We have expressed the incremental optimization problem as a shortest path problem on the time expanded network.  {The latter problem is formulated as follows: }
\begin{align}
\label{pro:time-exp}
  \begin{aligned}
	& \min	&&\sum_{k=0}^{K}~\sum_{(i,j)\in P^0} c_{ij}x_{ij}^k+\sum_{k=0}^{K}~\sum_{(i,j)\in A\setminus P^0} c_{ij}x_{ij}^k\\
          		  		& \text{~s.t.}   &&\begin{aligned}[t]
        							x^k_{ii}+\sum_{j:(i,j)\in P^0}x^k_{ij}+\sum_{j:(i,j)\in A\setminus P^0}x^k_{ij}~~~&\\  -x_{ii}^{k-1}-\sum_{j:(j,i)\in P^0}x^{k}_{ji}-\sum_{j:(j,i)\in A\setminus P^0}x^{k-1}_{ji}&=\begin{cases}
            1  & \text{if } i=s,k=0 \\
            -1 & \text{if } i= t,k=K \\
            0 & \text{otherwise} \\
          \end{cases}\quad ~\forall i\in N,~k=0,\ldots,K,\\
          x^k_{ij}&\in \{0,1\} \quad\quad \forall (i,j)\in P^0,~k=0,\ldots,K,\\
          x^k_{ij}&\in \{0,1\} \quad\quad \forall (i,j)\in A\setminus P^0,~ k=0,\ldots,K-1,\\
          x^k_{ii}&\in \{0,1\} \quad\quad \forall i\in N,~k=0,\ldots,K-1,\\
          x^K_{ij}&=0 \quad\quad\quad\quad \forall (i,j)\in A\setminus P^0,\\
          x^K_{ii}&=0\quad\quad\quad\quad \forall i\in N.
  \end{aligned}
\end{aligned}
\end{align}
In this problem, there are three subsets of decision variables corresponding to three subsets of arcs in the time expanded network. For every $(i,j)\in P^0$ and $k=0,\ldots,K$, the decision variable $x^k_{ij}$ corresponds to the arc from node $i_k$ to node $j_{k}$. For every $(i,j)\in A\setminus P^0$ and $k=0,\ldots,K-1$, the decision variable $x^k_{ij}$ corresponds to the arc from node $i_k$ to node $j_{k+1}$. For every $i\in N$ and $k=0,\ldots,K-1$, the decision variable $x^k_{ii}$ corresponds to the arc from node $i_k$ to node $i_{k+1}$. There is a one-to-one correspondence between feasible solutions of Problem~\eqref{pro:time-exp} and the paths from node $s_0$ to node $t_K$. We have included dummy variables $x_{ij}^K$ for $(i,j)\in A\setminus P^0$ and $x_{ii}^K$ for $i\in N$ to simplify the formulation, and then we let them to be zero.

 {We can relax the binary variables in Problem~\eqref{pro:time-exp}. By considering the dual problem, we obtain the following formulation for the incremental shortest path problem:}
\begin{align}
\label{pro:time-exp2}
  \begin{aligned}
	\opt_{\inc}(P^0,\bl{c}) = 	& \max	&&\pi_{s}^0-\pi_{t}^K \\
          		  		& \text{~s.t.}   &&\begin{aligned}[t]
        							\pi_{i}^k-\pi_{j}^k  &\leq c_{ij} && \forall (i,j)\in P^0,~k=0,\ldots,K,\\
        							\pi_{i}^k-\pi_{j}^{k+1}  &\leq c_{ij} && \forall (i,j)\in A\setminus P^0,~k=0,\ldots,K-1,\\
							\pi_{i}^k-\pi_{i}^{k+1}  &\leq 0 && \forall i\in N,~k=0,\ldots,K-1.
							  \end{aligned}
\end{aligned}
\end{align}

 {The adversarial problem is to find a cost vector $\bl{c} \in \mathcal{U}_1$ to maximize $\opt_{\inc}(P^0,\bl{c})$. This implies that the adversarial problem turns into a linear program as follows:}
\begin{align}
\label{pro:SPP-incopt}
\begin{aligned}
	\opt_{\adv}(P^0) =& \max	&&\pi_{s}^0-\pi_{t}^K\\
          		  		& \text{~s.t.}   &&\begin{aligned}[t]
        							\pi_{i}^k-\pi_{j}^k-\delta_{ij}  &\leq \bar{c}_{ij} && \forall (i,j)\in P^0,~k=0,\ldots,K,\\
        							\pi_{i}^k-\pi_{j}^{k+1}-\delta_{ij}  &\leq \bar{c}_{ij} && \forall (i,j)\in A\setminus P^0,~k=0,\ldots,K-1,\\
							\pi_{i}^k-\pi_{i}^{k+1}  &\leq 0 && \forall i\in N,~k=0,\ldots,K-1,\\
									 0\leq \delta_{ij}&\leq \hat{c}_{ij}  &&  \forall (i,j)\in A.
								\end{aligned}
\end{aligned}
\end{align}
 {Notice that this problem has $Kn+m$ variables and at most $(K+2)m+nK$ constraints. We assume without loss of generality that $K\leq n-1$ since otherwise the adversarial shortest path problem reduces to a nominal problem. This leads to the following result. }

\begin{theorem}
The adversarial shortest path problem under the uncertainty set $\mathcal{U}_1$ can be formulated as a linear program, and is solvable in polynomial time. 
\end{theorem}

Bertsimas and Sim \cite{BertsimasSim03} showed that the robust shortest path problem is solvable in polynomial time. We next prove that decision variant of the robust incremental shortest path problem is NP-complete. The transformation is from the \emph{2-disjoint-paths problem}, which is as follows. Given given a directed graph $G=(N,A)$ and distinct nodes $s_1, s_2, t_1, t_2$, are there two disjoint paths $P_1$ and $P_2$ such that $P_1$ is from $s_1$ to $t_1$ and $P_2$ is from $s_2$ to $t_2$?  This problem is shown to be NP-complete by Fortune \emph{et al.} \cite{Fortune80}.

\begin{theorem}
\label{thm:IR-SPP}
 {It is NP-hard to approximate the robust incremental shortest path problem within a factor of 2 if the uncertainty set is $\mathcal{U}_1$ and the distance metric is with respect to inclusion. }
\end{theorem}
\begin{proof}
Given an instance of the 2-disjoint-paths problem, we construct a graph $G'=(N',A')$ as follows. We create four nodes $s'_1, s'_2, t'_1, t'_2$, and  arcs $(s'_1,s_1)$, $(s'_1,t'_1)$, $(t_1,t'_1)$, $(s'_2,s_2)$, $(s'_2,t'_2)$, and $(t_2,t'_2)$. We furthermore link node $t'_1$ to node $s'_2$ with two parallel arcs. The construction of graph $G'$ is illustrated in Figure~\ref{fig:RI-SPP}.  
Suppose that the costs $d_{ij}$ for the initial path and the nominal arc costs $\bar{c}_{ij}$ are all zero, and suppose that an adversary can increase the cost of each arc by one.  Furthermore, let $\Gamma=K=1$.

We prove the following: (i) if the 2-disjoint-paths problem is a ``yes'' instance, then there exists a path $P^0$ from node $s'_1$ to node $t'_2$ in graph $G'$ with $\opt_{\adv}(P^0)=0.5$, where $\opt_{\adv}(P^0)$ is the optimal value of the adversarial optimization problem with respect to $P^0$; and (ii)  if the 2-disjoint-paths problem is a ``no'' instance, then $\opt_{\adv}(P^0) = 1$ for all paths $P^0$ from node $s'_1$ to node $t'_2$.   

We first suppose that there are two node-disjoint paths $P_1$ and $P_2$ from $s_1$ to $t_1$ and from $s_2$ to $t_2$, respectively. Let path $P^0$ be defined as follows:
\begin{align*}
P^0:=(s'_1,s_1),P_1,(t_1,t'_1), (t'_1,s'_2),(s'_2,s_2),P_2, (t_2,t'_s).
\end{align*}

If the adversary were to modify the cost of one arc only, then the best incremental path would have a cost of 0.  (This is easily seen by enumeration).  Thus, the adversary needs to modify the costs of at least two arcs.  An optimal choice for the adversary is to increase the costs of an arc of $P_1$ (or $P_2$) and the cost of the arc $(t'_1,s'_2)$ used in $P^0$, each with 0.5. Then an optimal incremental path would be to replace the arc $(t'_1,s'_2)$ used in $P^0$ by the other parallel arc $(t'_1,s'_2)$, which gives a cost of 0.5.

We now assume that the 2-disjoint-paths problem is a ``no'' instance and prove that $\opt_{\adv}(P^0)=1$ for every path $P^0$ from node $s'_1$ to node $t'_2$.  Consider an arbitrary path $P^0$. We consider first the case in which 
$P^0$ passes through node $t'_1$ (as well as node $s'_2$). In this case, it must contain arc $(s'_1,t'_1)$ or arc $(s'_2,t'_2)$.  We assume without loss of generality that it is arc $(s'_1,t'_1)$.  In this case, the adversary increases the cost of  arc $(s'_1,t'_1)$ by 1.   The incremental path must include this arc, and so has a cost of 1. 

Next, we assume that the path $P^0$ does not pass through node $t'_1$. In this case, path $P^0$ must contain the arc $(t_2,t'_2)$.   In this case, the adversary increases the cost of arc $(t_2,t'_2)$ by 1.  Every path $P$ from node $s'_1$ to node $t'_2$ with $|P\setminus P^0|\leq 1$ must contain this arc as well, which implies that $\opt_{\adv}(P^0)=1$.  This establishes the proof of the theorem.
  
\end{proof}

\begin{figure}[t]
    \centering
    \begin{tikzpicture}[inner sep=0.4mm]
  \node (1) at (0,0) [place] {$s'_1$};
  \node (2) at (1,0) [place] {$s_1$};
  \node (3) at (3,0) [place] {$i$};
  \node (4) at (5,0) [place] {$t_1$};
  \node (5) at (6,0) [place] {$t'_1$};
  \node (6) at (8,0) [place] {$s'_2$};
  \node (7) at (9,0) [place] {$s_2$};
  \node (8) at (11,0) [place] {$j$};
  \node (9) at (13,0) [place] {$t_2$};
  \node (10) at (14,0) [place] {$t'_2$};
  \node (11) at (1.8,0.5) [] {};
  \node (12) at (1.8,-0.5) [] {};  
  \node (13) at (2.2,0.5) [] {};
  \node (14) at (2.2,-0.5) [] {};  
  \node (15) at (3.8,0.5) [] {};
  \node (16) at (3.8,-0.5) [] {};  
  \node (17) at (4.2,0.5) [] {};
  \node (18) at (4.2,-0.5) [] {};  
    \node (19) at (9.8,0.5) [] {};
  \node (20) at (9.8,-0.5) [] {};  
  \node (21) at (10.2,0.5) [] {};
  \node (22) at (10.2,-0.5) [] {};  
  \node (23) at (11.8,0.5) [] {};
  \node (24) at (11.8,-0.5) [] {};  
  \node (25) at (12.2,0.5) [] {};
  \node (26) at (12.2,-0.5) [] {}; 
\begin{scope}[color=black,line width=1pt]
  \draw [->] (1)   --  (2) node [above, sloped,midway]{};
  \draw [->] (4)   --  (5) node [above, sloped,midway]{};
  \draw [->] (5)  to [bend right=45]   (6) node [above, sloped,midway]{};
  \draw [->] (5)   to [bend left=45]   (6) node [above, sloped,midway]{};
  \draw [->] (1)  to [bend left=65]   (5) node [above, sloped,midway]{};
  \draw [->] (6)   to [bend left=65]   (10) node [above, sloped,midway]{};
  \draw [->] (6)   --  (7) node [above, sloped,midway]{};
  \draw [->] (9)   --  (10) node [above, sloped,midway]{};
  \draw [->] (2)   --  (11) node [above, sloped,midway]{};
  \draw [->] (2)   --  (12) node [above, sloped,midway]{};
  \draw [->] (13)   --  (3) node [above, sloped,midway]{};
  \draw [->] (14)   --  (3) node [above, sloped,midway]{};
  \draw [->] (3)   --  (15) node [above, sloped,midway]{};
  \draw [->] (3)   --  (16) node [above, sloped,midway]{};
  \draw [->] (17)   --  (4) node [above, sloped,midway]{};
  \draw [->] (18)   --  (4) node [above, sloped,midway]{};
 \draw [->] (7)   --  (19) node [above, sloped,midway]{};
  \draw [->] (7)   --  (20) node [above, sloped,midway]{};
  \draw [->] (21)   --  (8) node [above, sloped,midway]{};
  \draw [->] (22)   --  (8) node [above, sloped,midway]{};
  \draw [->] (8)   --  (23) node [above, sloped,midway]{};
  \draw [->] (8)   --  (24) node [above, sloped,midway]{};
  \draw [->] (25)   --  (9) node [above, sloped,midway]{};
  \draw [->] (26)   --  (9) node [above, sloped,midway]{};
\end{scope}
\end{tikzpicture}
\caption{\label{fig:RI-SPP} Construction of $G'$ from $G$. The graph $G'$ is constructed from $G$ by introducing fours additional nodes $s'_1,s'_2,t'_1,t'_2$ and additional arcs $(s'_1,s_1)$,  $(s'_2,t'_2)$, $(s'_1,s_1)$,  $(t_2,t'_2)$ and two parallel arcs from $t'_1$ to $s'_2$. }
\end{figure}
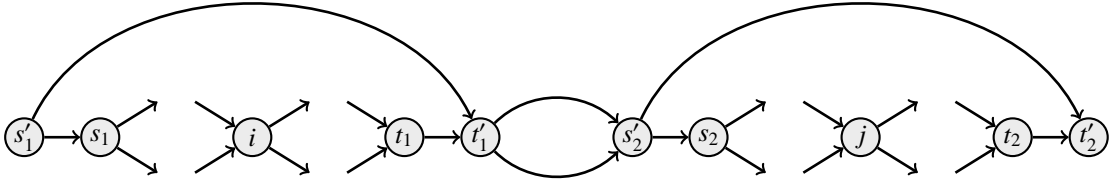

We have already mentioned that the robust incremental shortest path problem is NP-hard under the uncertainty set $\mathcal{U}_2$ since it includes as special case the problem of finding the  $K$ mostvital arcs. However, for a fixed $K$, and in particular $K=1$, the latter problem can be solved in polynomial time. The next result shows that the robust incremental shortest path problem is NP-hard with respect to the uncertainty set $\mathcal{U}_2$, even for $\Gamma=K=1$. 

\begin{theorem}
\label{thm:IR-SPP1}
The robust incremental shortest path problem is NP-hard under the uncertainty set $\mathcal{U}_2$. 
\end{theorem}
\begin{proof}
We use a similar reduction as in the proof of Theorem~\ref{thm:IR-SPP1} from the 2-disjoint-paths problem.
Consider an instance of the 2-disjoint-paths problem given by a directed graph $G=(N,A)$ and distinct nodes $s_1, s_2, t_1, t_2$. In addition, let $G'$ be the constructed graph as in the proof of Theorem~\ref{thm:IR-SPP1}.

We show that the 2-disjoint-paths problem is a ``yes'' instance if and only if there exists a path $P^0$ from $s'_1$ to $t'_2$ with $\opt_{\adv}(P^0)=0$.
We first suppose that there are two node-disjoint paths $P_1$ and $P_2$ from $s_1$ to $t_1$ and from $s_2$ to $t_2$, respectively. We set $P^0=(s'_1,s_1),P_1,(t_1,t'_1), (t'_1,s'_2),(s'_2,s_2),P_2, (t_2,t'_s)$.  If the cost of an arc which is not in the path $P^0$ is increased, then $P^0$ is still a path from $s'_1$ to $t'_2$ of cost zero. So we suppose that the cost of an arc in $P^0$ in increased by 1. If this arc is before node $t'_1$ (after node $s'_2$), then we add arc $(s'_1,t'_1)$ (arc $(s'_2,t'_2)$). In the case that the cost of arc $(t'_1,s'_2)$ is increased, the other parallel arc can be replaced.  Thus, in any case we can add one arc and have a new path from $s'_1$ to $t'_2$ of cost zero. 

We next proceed to prove the reverse direction.  Assume that there is path $P^0$ so that whenever the cost of one arc is increased, we can add at most one arc and have a path from $s'_1$ to $t'_2$ of cost zero. Such a path must not include either $(s'_1,t'_1)$ or  $(s'_2,t'_2)$ since otherwise when the cost of one of these two arcs increases, we cannot build up any path of cost zero by adding one arc. This implies that there is a path from $s_1$ to $t_1$ and a path from $s_2$ to $t_2$. This completes the proof of the theorem. 
  
\end{proof}

\subsection{The incremental function $|P\oplus P^0|$}

While the incremental shortest path problem with arc inclusion can be solved efficiently, the arc exclusion version is NP-complete \cite[see][]{SerefAhujaOrlin09}.  We next prove that the incremental shortest path problem is also NP-complete under the symmetric difference variant in which one is allowed to add or remove at most $K$ arcs from the given path. 

\begin{theorem}
Under the uncertainty set $\mathcal{U}_1$, the increment shortest path problem is NP-complete. 
\end{theorem}
\begin{proof}
The proof is based on a reduction from the 2-disjoint-paths problem. Consider an instance of the 2-disjoint-paths problem given by a directed graph $G=(N,A)$ and distinct nodes $s_1, s_2, t_1, t_2$. We construct a graph $G'=(N',A')$ as follows. For $i=1,2$, introduce an arc $(s_i,t_i)$. Furthermore, we link node $t_1$ to node $s_2$ with $n+1$ series arcs $(t_1,i_1),(i_1,i_2),\ldots,(i_{n-1},i_{n}), (i_{n},s_2)$. We associate a cost of zero to all original arcs in $A$ and a cost of one to the new arcs.   {The construction of graph $G'$ from graph $G$ is shown in Figure~\ref{fig:Red2}.}

\begin{figure}[t]
    \centering
    \begin{tikzpicture}[inner sep=0.4mm]
  \node (1) at (1,0) [place] {$s_1$};
  \node (2) at (3,0) [place] {$i$};
  \node (3) at (5,0) [place] {$t_1$};
  \node (4) at (6,0) [place] {$i_1$};
   \node (41) at (7,0)  {};
  \node (51) at (8,0) {};
  \node (5) at (9,0) [place] {$i_n$};
  \node (6) at (10,0) [place] {$s_2$};
  \node (7) at (12,0) [place] {$j$};
  \node (8) at (14,0) [place] {$t_2$};
  \node (9) at (1.8,0.5) [] {};
  \node (10) at (1.8,-0.5) [] {};  
  \node (11) at (2.2,0.5) [] {};
  \node (12) at (2.2,-0.5) [] {};  
  \node (13) at (3.8,0.5) [] {};
  \node (14) at (3.8,-0.5) [] {};  
  \node (15) at (4.2,0.5) [] {};
  \node (16) at (4.2,-0.5) [] {};  
   \node (17) at (10.8,0.5) [] {};
  \node (18) at (10.8,-0.5) [] {};  
  \node (19) at (11.2,0.5) [] {};
  \node (20) at (11.2,-0.5) [] {};  
  \node (21) at (12.8,0.5) [] {};
  \node (22) at (12.8,-0.5) [] {};  
  \node (23) at (13.2,0.5) [] {};
  \node (24) at (13.2,-0.5) [] {}; 
\begin{scope}[color=black,line width=1pt]
  \draw [->] (1)  to [bend left=65]   (3) node [above, sloped,midway]{};
   \draw [->] (6)   to [bend left=65]   (8) node [above, sloped,midway]{};
    \draw [->] (1)   --  (9) node [above, sloped,midway]{};
  \draw [->] (1)   --  (10) node [above, sloped,midway]{};
  \draw [->] (11)   --  (2) node [above, sloped,midway]{};
  \draw [->] (12)   --  (2) node [above, sloped,midway]{};
  \draw [->] (2)   --  (13) node [above, sloped,midway]{};
  \draw [->] (2)   --  (14) node [above, sloped,midway]{};
  \draw [->] (15)   --  (3) node [above, sloped,midway]{};
  \draw [->] (15)   --  (3) node [above, sloped,midway]{};
  \draw [->] (3)   --  (4) node [above, sloped,midway]{};
  \draw [->] (4)   --  (41) node [above, sloped,midway]{};
 \draw [->] (51)   --  (5) node [above, sloped,midway]{};
  \draw [->] (5)   --  (6) node [above, sloped,midway]{};
  \draw [->] (6)   --  (17) node [above, sloped,midway]{};
  \draw [->] (6)   --  (18) node [above, sloped,midway]{};
  \draw [->] (19)   --  (7) node [above, sloped,midway]{};
  \draw [->] (20)   --  (7) node [above, sloped,midway]{};
  \draw [->] (7)   --  (21) node [above, sloped,midway]{};
  \draw [->] (7)   --  (22) node [above, sloped,midway]{};
  \draw [->] (23)   --  (8) node [above, sloped,midway]{};
  \draw [->] (24)   --  (8) node [above, sloped,midway]{};
\end{scope}
\end{tikzpicture}
\caption{\label{fig:Red2}  {The graph $G'$ is constructed from $G$ by adding two new arcs $(s_1,t_1)$ and $(s_2,t_2)$ and by creating a path of length  $n+1$  from node $t_1$ to node $s_2$ with $n$ new nodes and $n+1$ new arcs.}}
\end{figure}
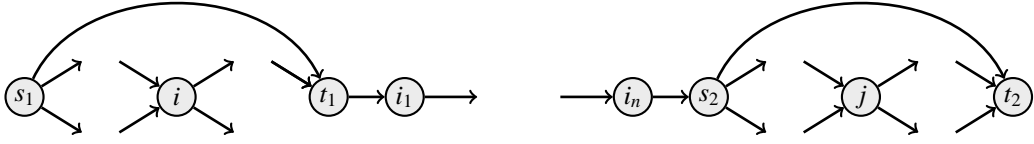

We now suppose that the path $P^0$ is given as 
\begin{align*}
P^0=(s_1,t_1), (t_1,i_1),(i_1,i_2),\ldots,(i_{n-1},i_{n}), (i_{n},s_1), (s_2,t_2),
\end{align*}
and let $K=n+1$. It is easy to see that the 2-disjoint-paths problem is a ``yes'' instance if and only if the optimal value of the incremental shortest path problem is $n+1$ because the incremental shortest path must include the arcs $(t_1,i_1),(i_1,i_2),\ldots, (i_{n},s_2)$. 
  
\end{proof}

\section{Adversarial Minimum Spanning Tree Problem}
\label{sec:MST}
In this section, we consider the \emph{minimum spanning tree} problem. 
In this problem, we are given an undirected graph $G=(N,A)$ with $n$ nodes and $m$ arcs. Each arc $(i,j)\in A$ has an associated \emph{cost}~$c_{ij}$. The problem is to determine a spanning tree of minimum cost. This problem can be formulated as follows:
\begin{align}
\label{pro:MST}
  \begin{aligned}
	& \min	&&\sum_{(i,j)\in A} c_{ij}x_{ij} \\
          		  		& \text{~s.t.}   &&\begin{aligned}[t]
        							\sum_{(i,j)\in A} x_{ij} &\geq n-1, \\
							\sum_{(i,j)\in A: i,j \in S} x_{ij}&\leq |S|-1 && \forall S\subseteq N,\\
							x_{ij}     &\in  \{0,1\} &&\forall (i,j)\in A.          
  \end{aligned}
\end{aligned}
\end{align}
The binary variable $x_{ij}$ indicates arc $(i,j)$ is in the tree or not, depending on whether $x_{ij}=0 $ or 1, respectively. The first constraint guarantees the existences of at least $n-1$ arcs.  The second set of constraints, referred to it as \emph{subtour-elimination constraints}, ensures that the set of selected arcs does not include any  cycle. There is a one-to-one correspondence between feasible solutions of Problem~\eqref{pro:MST} and spanning trees in graph $G$.  

The linear programming relaxation of Problem~\eqref{pro:MST} is known as  the \emph{subtour LP}.  Edmonds \cite{Edmonds67} shows that every extreme point of the feasible region of the subtour LP is binary and corresponds to the incidence vector of a spanning tree.  In addition, there is a polynomial time separation oracle for the constraints in the subtour LP~(see, e.g., \citep{Schrijver03}).  We will rely on this separation algorithm in order to efficiently solve the adversarial problem.

We suppose that each arc $(i,j)\in A$ has an associated interval $[\bar{c}_{ij},\bar{c}_{ij}+\hat{c}_{ij}]$.  Any possible realization of the arc cost $c_{ij}$ lies in this interval. Given an initial spanning tree $T^0$,  one can modify the tree after observing the arc costs. We measure the distance between two spanning trees by the arc inclusion operator; that is, one can build a new spanning tree by adding at most new $K$ arcs. Since each spanning tree has $n-1$ arcs, this case is equivalent to excluding at most $K$ arcs.  It is also equivalent to finding  a symmetric difference with at most $2K$ arcs.  From the viewpoint of complexity analysis, these three problem variants are equivalent.  Accordingly, we just consider the case where a new spanning tree can be built by adding $K$ new arcs.

We will examine the adversarial minimum spanning tree problem under the uncertainty sets $\mathcal{U}_1$ and $\mathcal{U}_2$. We first consider the uncertainty set $\mathcal{U}_2$. 
The adversarial problem is the problem of finding the $\Gamma$ most vital  arcs, i.e., the $\Gamma$ arcs whose removal from the graph will lead in the greatest increase in the cost of the minimum spanning tree in the remaining graph. Lin \emph{et al.} \cite{LinChern93} show that this problem is NP-hard for arbitrary $\Gamma$. Their result immediately implies the following.

\begin{theorem}
Under the uncertainty set $\mathcal{U}_2$, the decision variant of the adversarial minimum spanning tree problem is NP-complete.
\end{theorem}

We next show that the adversarial minimum spanning tree problem is solvable in polynomial time if the possible range of the cost vector $\bl{c}$ is given by $\mathcal{U}_1$. Given a spanning tree $T^0$ and a vector $\bl{c}$ of costs, we define $\opt_{\inc}(T^0,\bl{c})$ to be the optimal value of the incremental optimization problem, that is, 
\begin{align}
\label{pro:Inc-minimum spanning tree(NLP)}
  \begin{aligned}
	\opt_{\inc}(T^0,\bl{c}) := 	& \min	&&\sum_{(i,j)\in A} c_{ij}x_{ij} \\
          		  		& \text{~s.t.}   &&\begin{aligned}[t]
							\sum_{(i,j)\notin T^0} x_{ij}&\leq K,\\
        							\bl{x}&\in \mathcal{S}.\\
  \end{aligned}
\end{aligned}
\end{align}
 {where $\mathcal{S}$ is the feasible region of Problem~\eqref{pro:MST}. }

 {Problem \eqref{pro:Inc-minimum spanning tree(NLP)} has the integrality property because it is represented as the intersection of two matroids: a uniform matroid and the forest matroid \citep[see][]{Edmonds70}. We then apply the Lagrangian relaxation method to relax the first constraint.} By associating a Lagrangian multiplier with respect to the first constraint of Problem~\eqref{pro:Inc-minimum spanning tree(NLP)}, we obtain the following relaxation problem: 
\begin{align}
\label{pro:MST-L(lambda)}
	L(\lambda) := \min_{\bl{x}\in \mathcal{S}}  &\quad	 \sum_{(i,j)\in A} c_{ij}x_{ij}+\lambda \sum_{(i,j)\notin T^0}x_{ij}-\lambda K.
\end{align}
The binary variables $x_{ij}\in \{0,1\}$ can be replaced by $x_{ij}\geq 0$ in the Lagrangian relaxation $L(\lambda)$ (note that the constraints $x_{ij}\leq 1, (i,j)\in A$ become redundant due to the subtour-elimination constraints). By strong duality for linear programming, we can write
\begin{align}
\label{pro:MST-L(lambda)*}
\begin{aligned}
	L(\lambda) = 	& \max	&&(n-1)w-\sum_{S\subseteq N} (|S|-1)y_S-\lambda K \\
          		  		& \text{~s.t.}   &&\begin{aligned}[t]
									w-\sum_{S\subseteq N: i,j\in S} y_S&\leq c_{ij} && (i,j)\in T^0,\\
									w-\sum_{S\subseteq N: i,j\in S} y_S&\leq c_{ij}+\lambda && (i,j)\notin T^0,\\
        									y_S&\geq 0 && \forall S\subseteq N,\\
        									w&\geq 0.
								\end{aligned}
\end{aligned}
\end{align}
Because Problem (\ref{pro:Inc-minimum spanning tree(NLP)}) has the integrality property, 
$\opt_{\inc}(T^0,\bl{c})=\max_{\lambda\geq 0} L(\lambda)$. 
The adversarial problem seeks a cost vector $\bl{c} \in \mathcal{U}_1$ that maximizes $\opt_{\inc}(T^0,\bl{c})$. Let $\opt_{\adv}(T^0)$ denote the optimal value of the adversarial problem, starting from the initial tree $T^0$. Following the above discussion, we have
\begin{align}
\label{pro:MST-advopt}
\begin{aligned}
	\opt_{\adv}(T^0) = 	&\max	&&(n-1)w-\sum_{S\subseteq N} (|S|-1)y_S- \lambda K \\
          		  		& \text{~s.t.}   &&\begin{aligned}[t]
									w-\sum_{S\subseteq N: i,j\in S} y_S-\delta_{ij}&\leq \bar{c}_{ij} && (i,j)\in T^0,\\
									w-\sum_{S\subseteq N: i,j\in S} y_S-\delta_{ij}-\lambda&\leq \bar{c}_{ij} && (i,j)\notin T^0,\\
									\sum_{(i,j)\in A} \delta_{ij}&\leq \Gamma,\\
									0\leq \delta_{ij}&\leq \hat{c}_{ij}  && \forall (i,j)\in A,\\
        									y_S&\geq 0 && \forall S\subseteq N,\\
        									\lambda,w&\geq 0.
								\end{aligned}
\end{aligned}
\end{align}

The dual problem is:
\begin{align}
\label{pro:MST-advopt*}
\begin{aligned}
	\opt_{\adv}(T^0) = 	&\min	&&\sum_{(i,j)\in A} \bar{c}_{ij}x_{ij}+\sum_{(i,j)\in A} \beta_{ij}\hat{c}_{ij}+\theta\Gamma \\
          		  		& \text{~s.t.}   &&\begin{aligned}[t]
									-x_{ij}+\beta_{ij}+\theta&\geq 0 && \forall (i,j)\in A,\\
									\sum_{(i,j)\notin T^0} x_{ij}&\leq K,\\
									\beta_{ij}&\geq 0  && \forall (i,j)\in A,\\	
									\theta&\geq 0,\\	
        									\bl{x}&\in \mathcal{S}_{\text{R}},
								\end{aligned}
\end{aligned}
\end{align}
 {where $\mathcal{S}_{\text{R}}$ is the feasible region of the subtour LP.}

We notice that the set $\mathcal{S}_{\text{R}}$ contains exponentially many constraints. However, as mentioned before, there exists a polynomial time separation oracle for the constraints in $\mathcal{S}_{\text{R}}$. This implies that  Problem~\eqref{pro:MST-advopt*} can solved in polynomial time by the equivalence of optimization and separation \citep{GLS88}.  

\begin{theorem}
Under the uncertainty set $\mathcal{U}_1$, the adversarial minimum spanning tree problem can be solved in polynomial time.
\end{theorem}

\section{Conclusions}
We have presented a robust incremental approach to address uncertainty in optimization problems. We model the case in which the decision maker is allowed to make an incremental change after observing the realization of the uncertain parameters. We addressed the complexity of several optimization problems within this framework. We showed that the robust incremental counterpart of a linear programming problem is solvable in polynomial time. We also established NP-hardness of several robust incremental problems by showing that the adversarial problem or the incremental problem is NP-hard. This puts a lower bound on the complexity of  these problems. It is open as to whether the decision versions of the robust incremental problems are $\Sigma^2_p$-complete. We remark that if the incremental problem as well as the adversarial problem are in the class NP, then the robust incremental problem is in the class $\Sigma^2_p$. Another open issue is the complexity of the robust incremental minimum spanning tree problem.

\bibliographystyle{ormsv080}
\bibliography{mybib}

\begin{thebibliography}{29}
\expandafter\ifx\csname natexlab\endcsname\relax\def\natexlab#1{#1}\fi
\expandafter\ifx\csname url\endcsname\relax
  \def\url#1{{\tt #1}}\fi
\expandafter\ifx\csname urlprefix\endcsname\relax\def\urlprefix{URL }\fi
\expandafter\ifx\csname urlstyle\endcsname\relax
  \expandafter\ifx\csname doi\endcsname\relax
  \def\doi#1{doi:\discretionary{}{}{}#1}\fi \else
  \expandafter\ifx\csname doi\endcsname\relax
  \def\doi{doi:\discretionary{}{}{}\begingroup \urlstyle{rm}\Url}\fi \fi

\bibitem[{Bar-Noy et~al.(1995)Bar-Noy, Khuller, and
  Schieber}]{Bar-NoyKhullerSchieber95}
Bar-Noy, A, S.~Khuller, B~Schieber. 1995.
\newblock The complexity of finding most vital arcs and nodes.
\newblock Technical Report CS-TR-3539, Institute for Advanced Studies,
  University of Maryland, College Park, MD,.

\bibitem[{Ben-Tal et~al.(2009)Ben-Tal, El~Ghaoui, and Nemirovski}]{BenElNem09}
Ben-Tal, A., L.~El~Ghaoui, A.S. Nemirovski. 2009.
\newblock {\it Robust Optimization\/}.
\newblock Princeton Series in Applied Mathematics, Princeton University Press.

\bibitem[{Ben-Tal et~al.(2004)Ben-Tal, Goryashko, Guslitzer, and
  Nemirovski}]{BenTal04}
Ben-Tal, A., A.~Goryashko, E.~Guslitzer, A.~Nemirovski. 2004.
\newblock Adjustable robust solutions of uncertain linear programs.
\newblock {\it Mathematical Programming\/} {\bf 99}(2) 351--376.

\bibitem[{Ben-Tal and Nemirovski(1999)}]{Ben-TalNemirovski99}
Ben-Tal, A., A.~Nemirovski. 1999.
\newblock Robust solutions of uncertain linear programs.
\newblock {\it Operations Research Letters\/} {\bf 25} 1--13.

\bibitem[{Bertsimas et~al.(2011)Bertsimas, Brown, and
  Caramanis}]{BertsimasBrownCara11}
Bertsimas, D., D.~B. Brown, C.~Caramanis. 2011.
\newblock Theory and applications of robust optimization.
\newblock {\it SIAM Review\/} {\bf 53}(3) 464--501.

\bibitem[{Bertsimas and Sim(2003)}]{BertsimasSim03}
Bertsimas, D., M.~Sim. 2003.
\newblock Robust discrete optimization and network flows.
\newblock {\it Mathematical Programming\/} {\bf 98} 49--71.

\bibitem[{Bertsimas and Sim(2004)}]{BertsimasSim04}
Bertsimas, D., M.~Sim. 2004.
\newblock The price of robustness.
\newblock {\it Operations Research\/} {\bf 54} 35--53.

\bibitem[{Birge and Louveaux(1997)}]{BirgeLouveaux97}
Birge, J.~R., F.~Louveaux. 1997.
\newblock {\it Introduction to Stochastic Programming\/}.
\newblock Springer Series in Operations Research and Financial Engineering,
  Springer.

\bibitem[{B{\"u}sing(2012)}]{Busing12}
B{\"u}sing, C. 2012.
\newblock Recoverable robust shortest path problems.
\newblock {\it Networks\/} {\bf 59}(1) 181--189.

\bibitem[{\c{S}eref et~al.(2009)\c{S}eref, Ahuja, and
  Orlin}]{SerefAhujaOrlin09}
\c{S}eref, O., R.~K. Ahuja, J.~B. Orlin. 2009.
\newblock Incremental network optimization: Theory and algorithms.
\newblock {\it Operations Research\/} {\bf 57}(3) 586--594.

\bibitem[{Edmonds(1967)}]{Edmonds67}
Edmonds, J. 1967.
\newblock {Optimum Branchings}.
\newblock {\it Journal of Research of the National Bureau of Standards\/} {\bf
  71B} 233--240.

\bibitem[{Edmonds(1970)}]{Edmonds70}
Edmonds, J. 1970.
\newblock Submodular functions, matroids, and certain polyhedra.
\newblock {\it Combinatorial Structures and Their Applications\/}  69--87.

\bibitem[{Fortune et~al.(1980)Fortune, Hopcroft, and Wyllie}]{Fortune80}
Fortune, S., J.~E. Hopcroft, J.~Wyllie. 1980.
\newblock The directed subgraph homeomorphism problem.
\newblock {\it Theoretical Computer Science\/} {\bf 10} 111--121.

\bibitem[{Goetzmann et~al.(2011)Goetzmann, Stiller, and
  Telha}]{GoetStilleTelha11}
Goetzmann, K.-S., S.~Stiller, C.~Telha. 2011.
\newblock Optimization over integers with robustness in cost and few
  constraints.
\newblock Roberto Solis-Oba, Giuseppe Persiano, eds., {\it WAOA\/}, {\it
  Lecture Notes in Computer Science\/}, vol. 7164. Springer.

\bibitem[{Gr{\"{o}}tschel et~al.(1988)Gr{\"{o}}tschel, Lov{\'a}sz, and
  Schrijver}]{GLS88}
Gr{\"{o}}tschel, M., L.~Lov{\'a}sz, A.~Schrijver. 1988.
\newblock {\it Geometric {A}lgorithms and {C}ombinatorial {O}ptimization\/},
  {\it Algorithms and Combinatorics\/}, vol.~2.
\newblock Springer, Berlin.

\bibitem[{Israeli and Wood(2002)}]{IsraeliWood02}
Israeli, E., R.~K. Wood. 2002.
\newblock {Shortest-Path Network Interdiction}.
\newblock {\it Networks\/} {\bf 40}(2) 97--111.

\bibitem[{Kall and Wallace(1994)}]{KallWallace94}
Kall, P., S.~Wallace. 1994.
\newblock {\it Stochastic Programming\/}.
\newblock John Wiley $\&$ Sons.

\bibitem[{Liebchen et~al.(2009)Liebchen, L\"ubbecke, M\"ohring, and
  Stiller}]{Liebchen09}
Liebchen, C., M.~E. L\"ubbecke, R.~H. M\"ohring, S.~Stiller. 2009.
\newblock The concept of recoverable robustness, linear programming recovery,
  and railway applications.
\newblock R.K. Ahuja, R.H. M\"ohring, Chr. Zaroliagis, eds., {\it Robust and
  Online Large-Scale Optimization\/}, {\it LNCS\/}, vol. 5868. Springer-Verlag,
  Berlin, 1--27.

\bibitem[{Lin and Chern(1993)}]{LinChern93}
Lin, K.~C., M.~S. Chern. 1993.
\newblock The most vital edges in the minimum spanning tree problem.
\newblock {\it Information Processing Letters\/} {\bf 45}(1) 25--31.

\bibitem[{Meyer and Stockmeyer(1972)}]{meyer1972equivalence}
Meyer, A.R., L.J. Stockmeyer. 1972.
\newblock The equivalence problem for regular expressions with squaring
  requires exponential space.
\newblock {\it 13th Annual Symposium on Switching and Automata Theory\/}. IEEE,
  125--129.

\bibitem[{Rocco~Sanseverino and Ramirez-Marquez(2010)}]{SanseverinoMarquez10}
Rocco~Sanseverino, C.~M., J.~E. Ramirez-Marquez. 2010.
\newblock A bi-objective approach for shortest-path network interdiction.
\newblock {\it Computers and Industrial Engineering,\/} {\bf 59}(2) 232--240.

\bibitem[{Royset and Wood(2007)}]{RoysetWood07}
Royset, J.~O., R.~K. Wood. 2007.
\newblock Solving the bi-objective maximum-flow network-interdiction problem.
\newblock {\it INFORMS Journal on Computing\/} {\bf 19} 175--184.

\bibitem[{Schrijver(2003)}]{Schrijver03}
Schrijver, A. 2003.
\newblock {\it Combinatorial Optimization - Polyhedra and Efficiency\/}.
\newblock Springer.

\bibitem[{Shapiro et~al.(2009)Shapiro, Dentcheva, and Ruszczynski}]{Shapiro09}
Shapiro, A., D.~Dentcheva, A.~Ruszczynski. 2009.
\newblock {\it Lectures on Stochastic Programming: Modeling and Theory\/}.
\newblock No.~9 in MPS-SIAM series on optimization, SIAM, Philadelphia.

\bibitem[{Shen(1995)}]{Shen95}
Shen, H. 1995.
\newblock Finding the $k$ most vital edges with respect to minimum spanning
  tree.
\newblock {\it Acta Informatica\/} {\bf 36} 405--424.

\bibitem[{Stockmeyer(1976)}]{stockmeyer1976polynomial}
Stockmeyer, L.J. 1976.
\newblock The polynomial-time hierarchy.
\newblock {\it Theoretical Computer Science\/} {\bf 3}(1) 1--22.

\bibitem[{van Hoesel et~al.(1989)van Hoesel, Kolen, Kan, and
  Wagelmans}]{vanHoesel89}
van Hoesel, C. P.~M., A.~W.~J. Kolen, Rinnooy A. H.~G. Kan, A.~P.~M. Wagelmans.
  1989.
\newblock Sensitivity analysis in combinatorial optimization: A bibliography.
\newblock Technical report {8944/A}, Econometric Institute, Erasmus University,
  Rotterdam, The Netherlands.

\bibitem[{Wollmer(1964)}]{Wollmer64}
Wollmer, R. 1964.
\newblock Removing arcs from a network.
\newblock {\it Transportation Research~B\/} {\bf 12} 934--940.

\bibitem[{Wood(1993)}]{Wood93}
Wood, R.~K. 1993.
\newblock Deterministic network interdiction.
\newblock {\it Mathematical and Computer Modelling\/} {\bf 17} 1--18.

\end{thebibliography}

\end{document}